\documentclass[11pt, english]{article}
\usepackage[T1]{fontenc}
\usepackage[latin9]{inputenc}
\usepackage[margin=1in,bottom=1in,top=1in]{geometry}
\usepackage{lscape}
\usepackage{setspace}
\usepackage{amsmath}
\usepackage{amsthm}
\usepackage{bm}
\usepackage{amssymb}
\usepackage{stmaryrd}
\usepackage{graphicx}
\usepackage{float}
\usepackage{setspace}
\usepackage{rotating}
\usepackage{subcaption}
\usepackage{ragged2e}
\usepackage{esint}
\usepackage{ulem}
\usepackage{babel}
\usepackage{longtable}
\usepackage{booktabs}
\usepackage{caption,booktabs,array}

\usepackage[round]{natbib}
\usepackage[hypertexnames=false]{hyperref}

\usepackage{colortbl}
\definecolor{darkblue}{rgb}{0,0,.6}
\hypersetup{
colorlinks=true,
linkcolor=darkblue,
filecolor=darkblue,
urlcolor=darkblue,
citecolor=darkblue,
}
\usepackage[toc,page]{appendix}

\usepackage[sc]{mathpazo}
\usepackage{colortbl}

\newtheorem{assumption}{Assumption}

\newtheorem{remark}{Remark}
\newtheorem{example}{Example}
\newtheorem{theorem}{Theorem}

\newtheorem{proposition}{Proposition}
\newtheorem{lemma}{Lemma}
\newtheorem{corollary}{Corollary}

\usepackage{dsfont}
\usepackage{pifont}

\usepackage{enumerate}

\newcommand*{\latestversion}{\href{https://www.pietroemiliospini.com/publication/mte-misclassification/}{Click here for the latest version}} 
\newcommand*{\thisdraft}{This draft: November 2021} 
\newcommand*{\firstdraft}{First draft: June 2021}  

\begin{document}
\normalem

\title{MTE with Misspecification\thanks{We thank Vitor Possebom, Yixiao Sun, Kaspar Wuthrich, and seminar participants at UC San Diego for helpful comments. All remaining errors are ours. }}
\author{Juli\'an Mart\'inez-Iriarte \thanks{%
Email: jmart425@ucsc.edu} \\
Department of Economics\\
UC Santa Cruz \and  Pietro Emilio Spini \thanks{%
Email: pspini@ucsd.edu} \\
Department of Economics\\
UC San Diego }
\date{\latestversion \\
\thisdraft \\
\firstdraft}
\maketitle
\thispagestyle{empty}
\vspace{-2em}
\begin{abstract}\singlespacing
This paper studies the implication of a fraction of the population not responding to the instrument when selecting into treatment. We show that, in general, the presence of non-responders biases the Marginal Treatment Effect (MTE) curve and many of its functionals. Yet, we show that, when the propensity score is fully supported on the unit interval, it is still possible to restore identification of the MTE curve and its functionals with an appropriate re-weighting. 
\end{abstract}

\textbf{Keywords}: Marginal Treatment Effects, Misspecification, Weak Instruments.\bigskip

\clearpage
\pagenumbering{arabic} 

\section{Introduction}

Marginal treatment effects (MTEs) have unified the identification theory of several policy parameters. While the MTE framework is essentially non-parametric,\footnote{Linearity is sometimes assumed to facilitate estimation. See, e.g., Appendix B in \cite{urzua2006}} it is required that the recipient's participation into treatment follows a (generalized) Roy model. This is often referred to as additive separability: an ``additive'' comparison of costs and benefits determines selection. On the other hand, identification of the MTE is achieved via the local instrumental variable (LIV) approach (\cite{Heckman2001,Heckman2005}). An excellent survey is provided by \cite{mogstad2018}. 
An early effort to analyze MTE under misspecification can be found in the appendix of the seminal paper by \cite{Heckman2001}. They consider a case where the additive separability in the selection equation does not hold. The most serious consequence is that the LIV approach does not identify the MTE curve.  

In this paper we analyze a different type of misspecification. We model a situation in which, under additive separability, a proportion of the population does not take into account the instrumental variable when deciding whether to take up treatment or not. We refer to them as non-responders. To analyze the resulting bias, we define a pseudo-MTE curve which results from the LIV approach. Under no misspecification, the pseudo-MTE curve would coincide with the MTE curve. The resulting bias can be interpreted as a location-scale change of the MTE curve, parameterized by the proportion of non-responders and their propensity score.


We have two main results. The first one shows that the ability to recover the conditional average treatment effect (CATE) for the subpopulation of responders depends on the proportion of non-responders only through the support of the responders' propensity score. Indeed, when the support of the propensity score is the unit interval, it is possible to identify the CATE \emph{without} having to recover the true MTE curve in the first place. In a nutshell, ignoring misspecification and integrating under the pseudo-MTE curve over the support of observed propensity score yields the correct CATE for the subpopulation of responders. 

While the previous identification result for the CATE is independent of the proportion of non-responders,
this is not true of the MTE curve and other 
parameters derived from it such as LATE and MPRTE. However, in our second result, we show how to recover the MTE curve for responders by undoing the location-scale change induced by the presence of non-responders. The correction is based on an estimate of the support of the propensity score and requires only observable data. It gives an estimator of the policy parameter of interest that is simple to implement. Cases where the propensity score is fully supported are relevant in practice. For a recent example, see the survey approach of \cite{Briggs2020} the probability of having a child is supported on the full unit interval.  

Recently, \cite{kedadni2021} and \cite{vitor2021}
focus on the effect of measurement error in treatment status on the MTE curve. We complement such results by noting that a simple change to our setup can cover the case of misclassification. In a setting where treatment status is misclassified, the observed outcome is generated with the true treatment status. In our setting of misclassification, the observed outcome can be regarded as a mixture of responders and non-responders. The proportion of non-responders is analogous to the proportion of misreporters. Indeed, our results also hold if instead of having a fraction of non-responders, we have a fraction of misreporters.

Another consequence of the presence of non-responders in the sample is that the effect of the instrumental variable on the propensity score is attenuated. Motivated by this, 
we model a situation where the proportion of non-responders approaches 1, analogous to the setting of weak instruments of \cite{stockstaiger1997}. Thus, we can derive weak-instrument-like asymptotic distributions for the parameters derived from the MTE curve.

The rest of the paper is organized as follows: section \ref{sec:misc_and_mte} introduces the model; section \ref{sec:recover} contains the main identification results; section \ref{sec:bounds} provides bounds for the case where the propensity score is not fully supported in the unit interval; section \ref{sec:weak_iv} traces the connection to the weak IV literature; and section \ref{sec:conclusion} concludes. While this paper only deals with identification, we expect to extend our results to cover estimation and inference. 

\section{Misspecification and MTE}\label{sec:misc_and_mte}
In this section we introduce our model for misspecification in the MTE framework (\cite{Bjorklund1987}, \cite{Heckman2001,Heckman2005}). We analyze the consequences of misspecification from the identification point of view.
\subsection{The Model}

 We start with a general non-separable potential outcome model
\begin{align*}
Y(0)&=h_0(X,U_0),\\
Y(1)&=h_1(X,U_1),\\
Y&=D^*Y(1)+(1-D^*)Y(0),
\end{align*}
where $D^*$ is the observed treatment status, $X$ are observable covariates with support denoted by $\mathcal X$, and $\left\{Y(0),Y(1)\right\}, Y$ are potential and observed outcomes, respectively. The functions $h_0$ and $h_1$ are unknown. 

We model misspecification as a situation where there are two types of individuals: responders and non-responders. Responders select into treatment taking into account the incentives in Z. Their selection equation is given by $ D=\mathds 1\left\{ \mu(X,Z)\geq V\right\}$. On the other hand, non-responders do not react to incentives in Z at all. Their selection equation is given by $\tilde D=\mathds 1\left\{ \tilde \mu(X)\geq \tilde V\right\}$. Notice how $Z$ is not featured in $\tilde{\mu}(\cdot)$. For the non-responders, Z fails the relevance condition of the standard MTE model. 

Let $S$ be the latent status of an individual: $S=1$ for a responder and $S=0$ for a non-responder. The observed treatment status $D^*$ is given by:
\begin{align}\label{eq:mixture_d}
D^* = S \cdot D + (1-S) \cdot \tilde{D}.
\end{align}

We allow for the proportion of non-responders may vary with $X$. To this end, we define $\delta_X = \Pr(S=0|X)=\Pr(D^*= \tilde{D}|X)$. Thus, for every subpopulation with characteristics $X=x$ there is a proportion $\delta_x = \Pr(S=0|X=x) \in [0,1)$ of non-responders. We consider values where $\sup_{x \in \mathcal{X}} \delta_x < 1$ to avoid a situation where no-one responds to the instrumental variable.

\begin{remark}
We observe $Y$ according to $Y=D^*Y(1)+(1-D^*)Y(0)$, which is given by the actual choice $D^*$. If, instead, we have $Y=DY(1)+(1-D)Y(0)$, then we can interpret $D^*$ as a misclassified treatment status. In this case, all individuals decide according to $D=\mathds 1\left\{ \mu(X,Z)\geq V\right\}$, but a fraction of them reports according to $ \tilde D=\mathds 1\left\{ \tilde \mu(X)\geq \tilde V\right\}$ See \cite{kedadni2021} and \cite{vitor2021} for recent studies on MTE under misclassification. 
\end{remark}

The econometrician observes a cross section of $(Y_i, D^*_i, X_i, Z_i)$. When $\delta_X=0$ almost surely, then $D^*=D$ and we are in the familiar MTE framework of \cite{Heckman2001,Heckman2005}. Otherwise, if $\delta_X \neq 0$ almost surely, for an observation of $D^*_i$, we do not know whether we are observing the treatment status of a non-responder or of a responder. That is, it is unknown if we are observing $D_i$ or $\tilde D_i$. 

\begin{assumption} \textbf{Type Independence.}  \label{Assumption_type} $S\perp Z\| X$.
\end{assumption}

Assumption \ref{Assumption_type} states that once we control for $X$, the latent status of a individuals does not vary with the instrumental variable Z. 

\begin{assumption}
\textbf{Relevance and Exogeneity} \label{Assumption_heckman}
\begin{enumerate}%
\item \label{relevance}$\mu(X,Z)$ is a nondegenerate random variable
conditional on $X$.
\item \label{exogeneity} $(U_{0},U_{1},V, \tilde V)$ are independent of $Z$
conditional on $X$. 
\end{enumerate}
\end{assumption}

Note that, for the subpopulation of non-responders, the instrument is valid but totally irrelevant. The larger the value of $\delta_x$, the ``weaker'' the instrument $Z$, since most participants with $X=x$ are non-responders. With the exception of the requirement that $\tilde V\perp Z\| X$, these are the same conditions of \cite{Heckman2001,Heckman2005}. Our additional requirement covers the subpopulation of non-responders: neither the ``cost'' of treatment $\tilde V$ nor the ``benefit'' $\tilde \mu(X)$ depend on $Z$ when conditioned on $X$.

\begin{example}
To fix ideas, we can think of a two part cost of providing the incentive. A fixed cost associated to targeting a particular subpopulation with covariates $X=x$ and the cost of the incentive itself. If Z is a voucher, there could be administrative costs associated to making it available to subpopulation $X=x$. For non-responders who do not redeem the voucher, the cost of the incentive is zero. Such a scenario would satisfy Assumption \ref{Assumption_heckman}.
\end{example}

The misclassification structure of Equation \eqref{eq:mixture_d} allows to define three different propensity scores. An observed/identified one which is based on the observables $(D^*,X,Z)$, and two latent/unobserved propensity scores: one for the reponders and one for the non-responders. Formally, they are given by
\begin{align*}
P^*(X,Z)&:=\Pr(D^*=1|X,Z) & \textbf{(Observed)}\\
P(X,Z)&:=\Pr(D=1|S=1,X,Z) & \textbf{(Responders)} \\
\tilde P(X)&:=\Pr(\tilde D=1|S=0,X) & \textbf{(Non-responders)}
\end{align*}
The next result takes (mainly) advantage of Assumption \ref{Assumption_type} to derive a useful linear relation between them.
\begin{lemma}\label{lemma_prop_score}
Under Assumptions \ref{Assumption_type} and \ref{Assumption_heckman}.\ref{exogeneity} we can relate the different propensity scores by 
\begin{align}\label{eq:obs_ps}
P^*(X,Z)=(1-\delta_X)\cdot P(X,Z)+\delta_X \cdot\tilde P(X).
\end{align}
\end{lemma}
\begin{proof}
Starting with the model in \eqref{eq:mixture_d} we can write
\begin{align*}
\Pr(D^*=1|X,Z)&=\Pr(S=1|X,Z)\cdot \Pr(D=1|S=1,X,Z)\\& + \Pr(S=0|X,Z)\cdot \Pr(\tilde D=1|S=0,X,Z).
\end{align*}
Assumption \ref{Assumption_type} simplifies the mixing probabilities to $\Pr(S=1|X)=1-\delta_X$ and $\Pr(S=0|X)=\delta_X$. We obtain
\begin{align*}
\Pr(D^*=1|X,Z)=(1-\delta_X)\cdot \Pr(D=1|S=1,X,Z) + \delta_X\cdot \Pr(\tilde D=1|S=0,X,Z).
\end{align*}
To see that $Pr(\tilde D=1|S=0,X,Z)=Pr(\tilde D=1|S=0,X)$, we note that By Assumptions \ref{Assumption_type} and \ref{Assumption_heckman}.\ref{exogeneity}: 
\begin{align*}
 \Pr(\tilde D=1|S=0,X,Z) =  \Pr(\tilde \mu(X)\geq\tilde V|S=0,X,Z)=\Pr(\tilde \mu(X)\geq\tilde V|X)= \Pr(\tilde D=1|S=0,X).
\end{align*}
Therefore
\begin{align*}
\Pr(D^*=1|X,Z)&=(1-\delta_X)\cdot \Pr(D=1|S=1, X,Z) + \delta_X\cdot \Pr(\tilde D=1|S=0,X)\\
&=(1-\delta_X)\cdot P(X,Z)+\delta_X \cdot\tilde P(X).
\end{align*}
\end{proof}

For a fixed $X=x$, the result in Lemma \ref{lemma_prop_score} shows that the observed propensity (still random through $Z$) is a linear transformation of the propensity score for the responders. If, additionally, we take two different values of $Z$, for example $z$ and $z'$, we can remove the contribution of $\tilde P(X)$, which is invariant with respect to $z$ and obtain\footnote{We write $P^*(x,z)$ for $\Pr(D^*=1|X=x,Z=z)$, and $P(x,z)$ for $\Pr(D=1|S=1,X=x,Z=z)$.}
\begin{align}\label{eq:disc_pz}
P^*(x,z)-P^*(x,z')=(1-\delta_x)\cdot \left[P(x,z)-P(x,z')\right]
\end{align}
Equation \eqref{eq:disc_pz} says that the changes on the observed propensity score induced by varying $Z$ are proportional to the changes on the true propensity score induced by varying $Z$. Thus, if we knew $\delta_x$, we could recover the change in the propensity score for the responders. When $Z$ is continuous, we can take a limiting version of this argument, \textit{e.g.}, as $z'\to z$, to obtain
\begin{align}\label{eq:der_pz}
\frac{\partial P^*(x,z)}{\partial z}=(1-\delta_x)\cdot \frac{\partial P(x,z)}{\partial z}.
\end{align}
Both the discrete (equation \eqref{eq:disc_pz}), and the continuous (equation\eqref{eq:der_pz}) change in the propensity score play a role in the relationship between the MTE curve (defined below) and certain parameters of interest.

\subsection{The MTE for Responders}
For the subpopulation of responders, the standard MTE framework holds. This motivates us to define an MTE curve for this subpopulation. In doing so, we are implicitly assuming that this is our object of interest. The reason for this is that many times we can also control the instrumental variable $Z$. Thus, to asses the effects of manipulations of $Z$ we look at the MTE curve for responders.

Let $\mathcal P_x$ and $\mathcal P^*_x$ denote the support of $P(x,Z):=\Pr(D=1|X=x,Z)$ and $P^*(x,Z):=\Pr(D^*=1|X=x,Z)$ respectively. For the subpopulation of responders, we rewrite the selection equation as $D=\mathds 1\left\{ P(X,Z)\geq U_D \right\}$ where $U_D\sim U_{(0,1)}$.\footnote{This follows from $D=\mathds 1\{ F_{V|S,X,Z}(\mu(X,Z)|1,X,Z)\geq F_{V|S,X,Z}(V|1,X,Z)\}$. Noting that by assumptions \ref{Assumption_heckman}.(\ref{exogeneity}) and \ref{Assumption_type}, we have $D=\mathds 1\{ P(X,Z)\geq F_{V|S,X}(V|1,X)\}$. Finally, we take $U_D:=F_{V|S,X}(V|1,X)$.} Thus, we define the MTE curve for responders as
$$\text{MTE}(u,x):=\mathbb E\left[Y(1)-Y(0)|S=1,U_D=u,X=x\right].$$ By the LIV approach we have the following equivalence result:\footnote{See \cite{Heckman2001} for sufficient conditions.}
\begin{equation}\label{eq:mte_liv}
\text{MTE}(u,x)=\frac{\partial \mathbb E\left[Y|S=1,P(X,Z)=u,X=x\right]}{\partial u}\text{ for }u\in \mathcal P_x.
\end{equation}
Since we do not observe $P(X,Z)$, this is \emph{not} an identification result in our setting. In a similar fashion, we \emph{define} the following pseudo-MTE curve:
\begin{align}\label{eq:pseudo_mte}
\text{MTE}^*(u,x;\delta_x):=\frac{\partial \mathbb E\left[Y|P^*(X,Z)=u,X=x\right]}{\partial u} \text{ for }u\in \mathcal P^*_x.
\end{align}

We emphasize that the pseudo-MTE curve is indexed by $\delta_x$ because it depends implicitly on the proportion of the nonresponders. From the data only, we can only compute $\text{MTE}^*(u,x;\delta_x)$, not $\text{MTE}(u,x)$. The pseudo-MTE curve is the curve that would be mistakenly taken to be the MTE curve. Indeed, in the absence of non-responders, $\text{MTE}^*(u,x;0)=\text{MTE}(u,x)$. If non-responders are present in the $X=x$ subpopulation, that is if $\delta_x> 0$, the observed $\text{MTE}^*(u,x;\delta_x)$ does not identify $\text{MTE}(u,x)$. In another words, the LIV approach is biased.
We can now fully characterize the bias induced by $\delta_x$ on the MTE curve. 
\begin{lemma}\label{lemma:bias_mte}
Under Assumptions \ref{Assumption_type} and \ref{Assumption_heckman}, we can write
\begin{align}\label{eq:equivalence_2}
\text{MTE}(v,x)=(1-\delta_x) \text{MTE}^*\left ( (1-\delta_x)v+\delta_x\tilde P(x),x;\delta_x\right ) \text{ for }v\in \mathcal P_x.
\end{align}
\end{lemma}

\begin{proof}
Using \eqref{eq:obs_ps}, for $u\in \mathcal P^*_x$, we can write
\begin{align*}
\mathbb E\left[Y|P^*(X,Z)=u,X=x\right] &= \mathbb E\left[Y|(1-\delta_x)\cdot P(X,Z)+\delta_x\cdot\tilde P(X)=u,X=x\right] \\
&= \mathbb E\left[Y\bigg | P(X,Z)=\frac{u-\delta_x \tilde P(x)}{1-\delta_x},X=x\right] 
\end{align*}
Differentiating with respect to $u$, we obtain
\begin{align}\label{eq:mte_mte}
\text{MTE}^*(u,x;\delta_x)=\frac{1}{1-\delta_x} \text{MTE}\left (\frac{u-\delta_x \tilde P(x)}{1-\delta_x},x\right ) \text{ for }u\in \mathcal P^*_x.
\end{align}
since $\frac{u-\delta_x \tilde P(x)}{1-\delta}\in \mathcal P_x$ by \eqref{eq:obs_ps}. Alternatively, we can write
\begin{align*}
\text{MTE}(v,x)=(1-\delta_x) \text{MTE}^*\left ( (1-\delta_x)v+\delta_x\tilde P(x),x;\delta_x\right ) \text{ for }v\in \mathcal P_x.
\end{align*}
\end{proof}

Lemma \ref{lemma:bias_mte} shows that the bias is in the form of both location and scale. Equation \eqref{eq:mte_mte}, which is equivalent to Equation \eqref{eq:equivalence_2},\footnote{Note the changes in the domain of integration between \eqref{eq:equivalence_2} and \eqref{eq:mte_mte}.} shows that $\text{MTE}^*$ is obtained by changing the location from $u$ to $u-\delta_x\tilde P(x)$, and rescaling by $(1-\delta_x)^{-1}$. Thus, as in a location-scale family of densities, we can regard $\text{MTE}^*$ as a family of curves, defined over $\mathcal P^*_x$, which is indexed by $\delta_x$ and $\tilde P(x)$. 

\section{Automatic and explicit de-biasing}\label{sec:recover}
We now introduce our two main results. We show that, for any subpopulation $X=x$ where the instrument is strong enough to induce a propensity score supported on the full unit interval $[0,1]$, the associated $CATE(x)$ can be identified for responders. This is true even if the $MTE^*(u,x,\delta_x)$ curve is biased for $MTE(u,x)$. We note that the identified $CATE(x)$ parameters corresponds to the subpopulation of responders. 

\begin{assumption}\label{full_support}\textbf{Full Support.} The support of $P(x,Z)$ is $\mathcal P_x=[0,1]$ for every $x$ in a subset $\mathcal{X}_B \subseteq \mathcal{X}$.
\end{assumption}

Assumption \ref{full_support} says that the incentive in the instrument $Z$ is strong enough to induce any individual in the $X=x$ subpopulation into or out of treatment.
Perhaps surprisingly, the $\text{CATE}(x)$, can be recovered 
only by resorting to the full support assumption. That is, to correctly compute the $\text{CATE}(x)$ we do not need to recover the true MTE curve for responders.
\begin{theorem}\label{theorem:cate}
Let Assumptions \ref{Assumption_type}, \ref{Assumption_heckman}, and \ref{full_support} hold. Then, for any $x \in \mathcal{X}_B$:
\begin{align*}
\text{CATE}(x) =\int_{\inf \mathcal P^*_x}^{\sup \mathcal P^*_x} \text{MTE}^*(u,x;\delta_x)du.
\end{align*}
\end{theorem}
\begin{proof}
The Conditional Average Treatment Effect, $\text{CATE}(x)$, could be computed using the true MTE curve (if it was observed) as
\begin{align*}
\text{CATE}(x) =\int_0^1 \text{MTE}(u,x)du.
\end{align*}
Given that $\mathcal P_x=[0,1]$, then $\mathcal P^*_x:= [\underline{p_x^*} , \overline{p_x^*}]$ where $\underline{p_x^*}:=\inf \mathcal P^*_x=\delta_x\tilde P(x)$ and $\overline{p_x^*}:=\sup \mathcal P^*_x(1-\delta_x)+\delta_x\tilde P(x)]$. Consider the integrating the pseudo-MTE curve over the support of the observed propensity score:
\begin{align*}
\int_{\delta_x\tilde P(x)}^{(1-\delta_x)+\delta_x\tilde P(x)} \text{MTE}^*(u,x;\delta_x)du.
\end{align*}
Using \eqref{eq:mte_mte}, we have
\begin{align*}
\int_{\delta_x\tilde P(x)}^{(1-\delta_x)+\delta_x\tilde P(x)} \text{MTE}^*(u,x;\delta_x)du &=
\int_{\delta_x\tilde P(x)}^{(1-\delta_x)+\delta_x\tilde P(x)}  \frac{1}{1-\delta_x} \text{MTE}\left (\frac{u-\delta_x \tilde P(x)}{1-\delta_x},x\right )   du\\
&=\int_{0}^{1}\text{MTE}(u,x)du\\
&=\text{CATE}(x)
\end{align*}
where we have done the change of variables
\begin{align*}
v = \frac{u-\delta_x \tilde P(x)}{1-\delta_x}.
\end{align*}
\end{proof}
\begin{remark}
The result of Theorem \ref{theorem:cate} states that by integrating the observed (and biased) marginal treatment effect curve over the support of the observed (and biased) propensity score leads to the $\text{CATE}(x)$ provided that the propensity score for responders has full support. Thus, under the type of misspecification described in \eqref{eq:mixture_d}, $\text{CATE}(x)$ is robust to $\delta_x\neq 0$.
\end{remark}

\begin{remark}
This result also hold in a setting of misclassification and was our original motivation. That is, in a setting where
instead of $Y=D^*Y(1)+(1-D^*)Y(0)$, we have $Y=DY(1)+(1-D)Y(0)$ and we interpret $D^*$ as a misclassified treatment status. 
\end{remark}
Unfortunately, the automatic ``de-biasing" in Theorem \ref{theorem:cate} does not hold for the other policy parameters that can be obtained via the MTE curve. 
On the other hand, we show that the full support assumption can be used to identify $\delta_x$ which allows an explicit ``de-biasing" procedure. 
Given that $\mathcal P^*_x:= [\underline{p_x^*} , \overline{p_x^*} ] =[\delta_x\tilde P(x), (1-\delta_x)+ \delta_x\tilde P(x)]$ we can actually identify both $\delta_x$ and $\tilde P(x)$. It follows then from Lemma \ref{lemma:bias_mte} that we can recover the $\text{MTE}(u,x)$ curve. 

\begin{proposition}
\label{prop:identification of delta_x}
Let Assumptions \ref{Assumption_type}, \ref{Assumption_heckman}, and \ref{full_support} hold. Then $\delta_x$ is identified for any $x \in \mathcal{X}_B$ through:
\begin{equation*}
\delta_x = 1-(\overline{p_x^*} -\underline{p_x^*} )
\end{equation*}
\end{proposition}
\begin{proof}
According to Equation \eqref{eq:obs_ps}, the range of the observed propensity score is given by $\mathcal P^*_x=[\delta_x\tilde P(x), (1-\delta_x)+ \delta_x\tilde P(x)]$. For each $x$, the observed propensity score $P^*(\cdot)$ can be viewed as an affine function of $P(\cdot)$. This affine function is parameterized by $\delta_x$ and $\tilde{P}_x$. For the endpoints $\underline{p}_x$ and $\overline{p}_x$ of the true propensity score, we have the mappings:
\begin{align*}
\underline{p_x} \mapsto (1-\delta_x) \underline{p_x} + \delta_x \tilde{P}(x) \\
\overline{p_x} \mapsto (1-\delta_x) \overline{p_x} + \delta_x \tilde{P}(x) 
\end{align*}
The images of this collection of mapping are observed. They are the endpoints of the observed propensity score $P^*(Z,x)$. If the original endpoints of the true $P(\cdot)$ are known to be $\underline{p_x} = 0$ and $\overline{p_x} =1$, like stated in Assumption \ref{full_support}, the mapping above can be recovered by the following system of two equations in two unknowns: $\tilde P(x)$ and $\delta_x$.
\begin{align*}
\underline{p_x^*} &= \delta_x\tilde P(x)\\
\overline{p_x^*} &=(1-\delta_x)+ \delta_x\tilde P(x)
\end{align*}
which implies that 
\begin{align*}
\delta_x &=1-(\overline{p_x^*} -\underline{p_x^*}) \\
\tilde P(x)&=\underline{p_x^*} \cdot \frac{1}{\delta_x}
\end{align*}
\end{proof}
The intuition for this result is simple. Because the original propensity score $P(Z,x)$, for any fixed $x$, is supported on the unit interval, the observed support $\mathcal P^*_x=[\underline{p_x^*} , \overline{p_x^*} ]$ will contain enough information to identify $\delta_x$. This is summarized Figure \ref{fig:MTE_graph}.

\begin{figure}
    \centering
    \includegraphics[width=\textwidth]{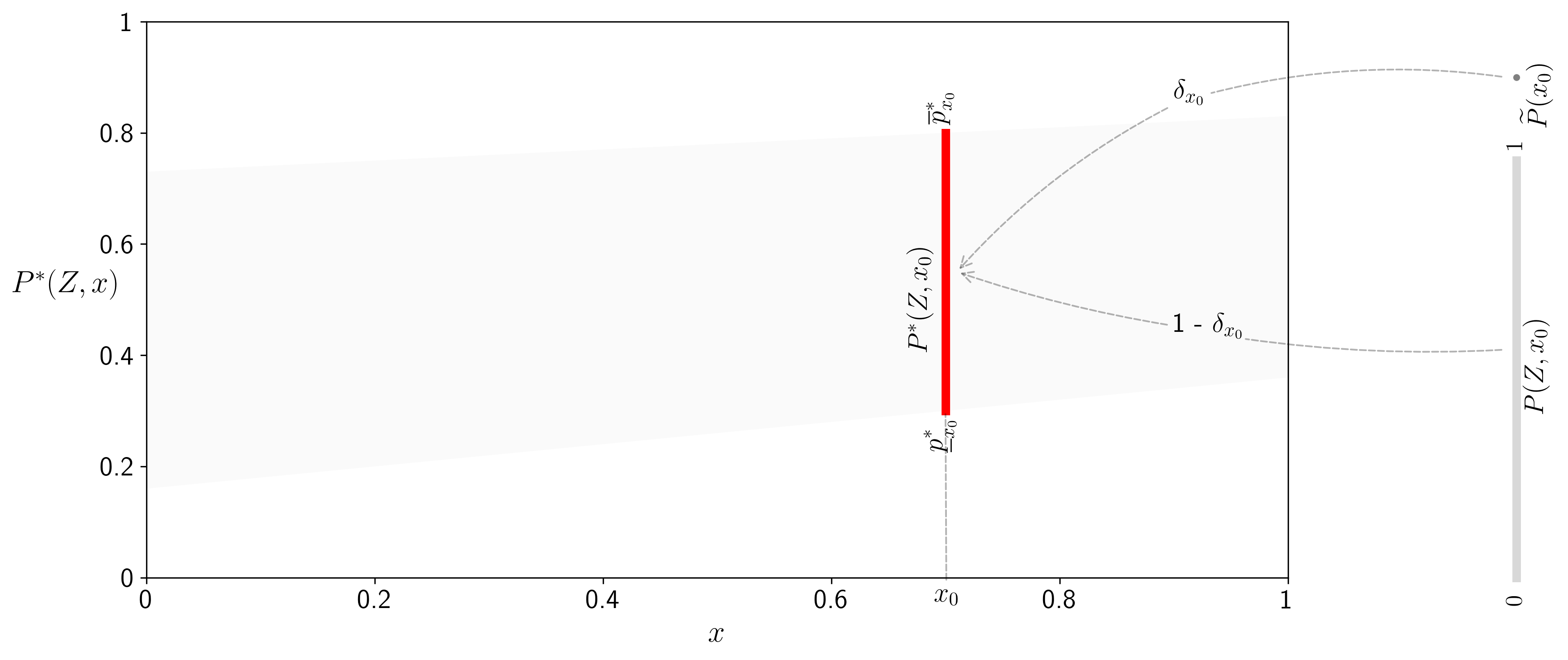}
    \caption{Identifying $\delta_x$: The figure shows the link between the non-responders propensity score, the proportion of non-responders and the observed propensity score. Because the non-responders propensity score does not vary with the instrument $Z$ and $supp(P(Z,x)) =[0,1]$ the $\delta_x$ can be recovered from observing the discrepancy from the observed support $P^*(Z,x)$ and $[0,1]$. The picture shows one of those points, $x_0$. }
    \label{fig:MTE_graph}
\end{figure}
Having identified $\delta_x$, then we use Equation \eqref{eq:mte_mte} to identify the \text{MTE} curve.

\begin{corollary}\label{cor:mte_id}
Let Assumptions \ref{Assumption_type}, \ref{Assumption_heckman}, and \ref{full_support} hold. Then, the \text{MTE} curve is identified:
\begin{align*}
\text{MTE}(v,x)=(\overline{p_x^*} -\underline{p_x^*} ) \text{MTE}^*\left ( (\overline{p_x^*} -\underline{p_x^*} )v+\underline{p_x^*},x;1-(\overline{p_x^*} -\underline{p_x^*})\right ) \text{ for }v\in \mathcal P_x = [0,1].
\end{align*}
where $\underline{p_x^*} = \inf \mathcal P^*_x$ and $ \overline{p_x^*} = \sup \mathcal P^*_x.$
\end{corollary}
This corollary provides the correct ``de-biasing'' to be performed on the observed MTE curve to match the true MTE curve. However, it is possible to recover parameters that are based on the MTE curve \emph{without} having to recover the MTE curve in the first place. We provide two examples.
\begin{example}[LATE]\label{example_late}
Consider the $\text{LATE}$, for $P(x,z')<P(x,z)$ with $z,z'\in\mathcal Z$, which can be obtained from MTE curve as
\begin{align*}
\text{LATE}(x, P(x,z), P(x,z'))=\frac{1}{P(x,z)-P(x,z')}\int_{P(x,z')}^{P(x,z)}\text{MTE}(u,x)du.
\end{align*}
Under misspecification, for the same $z,z'\in\mathcal Z$, we have
\begin{align*}
\text{LATE}^*(x, P^*(x,z), P^*(x,z'))&=\frac{1}{P^*(x,z)-P^*(x,z')}\int_{P^*(x,z')}^{P^*(x,z)}\text{MTE}^*(u,x;\delta_x)du\\
&=\frac{(1-\delta_x)^{-1}}{P(x,z)-P(x,z')}\int_{(1-\delta_x)P(x,z')+\delta_x\tilde P(x)}^{(1-\delta_x)P(x,z)+\delta_x\tilde P(x)}\frac{1}{1-\delta_x}\\ 
&\times \text{MTE}\left(\frac{u-\delta_x\tilde P(x)}{1-\delta_x},x\right)du.
\end{align*}
Note that to go from $\text{MTE}^*$ to $\text{MTE}$ we used Lemma \ref{lemma:bias_mte}. We did not use Corollary \ref{cor:mte_id}. Defining the change of variables $\tilde u = \frac{u-\delta_x\tilde P(x)}{1-\delta_x}$, we get $(1-\delta_x)d\tilde u = du.$ We then write
\begin{align*}
\text{LATE}^*(x, P^*(x,z), P^*(x,z'))&=\frac{(1-\delta_x)^{-1}}{P(x,z)-P(x,z')}\int_{(1-\delta_x)P(x,z')+\delta_x\tilde P(x)}^{(1-\delta_x)P(x,z)+\delta_x\tilde P(x)}\frac{1}{1-\delta_x}\\ &\times \text{MTE}\left(\frac{u-\delta_x\tilde P(x)}{1-\delta_x},x\right)du\\
&=\frac{(1-\delta_x)^{-1}}{P(x,z)-P(x,z')}\int_{P(x,z')}^{P(x,z)} \text{MTE}(u,x)du\\
&=\frac{1}{1-\delta_x}\text{LATE}(x, P(x,z), P(x,z')).
\end{align*}
Now, since $\delta_x =1-(\overline{p_x^*} -\underline{p_x^*})$ by Proposition \ref{prop:identification of delta_x}, the explicit de-biasing is achieved by
\begin{align*}
(\overline{p_x^*}-\underline{p_x^*})\text{LATE}^*(x, P^*(x,z), P^*(x,z'))&=\text{LATE}(x, P(x,z), P(x,z')).
\end{align*}
The left hand side can be computed from the data.
\end{example}

\begin{example}[MPRTE]\label{example_mprte}
The marginal policy relevant treatment effect (MPRTE) is an average of the $\text{MTE}(u,x)$ along the margin of indifference: when $U_D=P(X,Z)$. It is given by
\begin{align*}
\text{MPRTE}(x) = \int_{\mathcal Z} \text{MTE}(P(x,z),x)\frac{\partial P(x,z)}{\partial z} \left(E\left[\frac{\partial [P(x,Z)]}{\partial z}\right]\right)^{-1}f_{Z|X}(z|x)dz
\end{align*}
Then, using Equations \eqref{eq:der_pz} and \eqref{eq:equivalence_2} we get
\begin{align*}
\text{MPRTE}^*(x) &= \int_{\mathcal Z} \text{MTE}^*(P^*(x,z),x;\delta_x)\frac{\partial P^*(x,z)}{\partial z} \left(E\left[\frac{\partial [P^*(x,Z)]}{\partial z}\right]\right)^{-1}f_{Z|X}(z|x)dz\\
&= \int_{\mathcal Z} \frac{1}{1-\delta_x} \text{MTE}(P(x,z),x)\frac{\partial P(x,z)}{\partial z} \left(E\left[\frac{\partial [P(X,Z)]}{\partial z}\right]\right)^{-1}f_{Z|X}(z|x)dz\\
&=\frac{1}{1-\delta_x}\text{MPRTE}(x).
\end{align*}
Thus, again, by Proposition \ref{prop:identification of delta_x}, we obtain
\begin{align*}
(\overline{p_x^*}-\underline{p_x^*})\text{MPRTE}^*(x) &=\text{MPRTE}(x) .
\end{align*}

\end{example}

In the previous examples, proceeding as if there were no misspecification, yields biased parameters. Thus, the automatic ``de-biasing'' in CATE is the exception rather than the rule.

\section{Bounds under limited support}\label{sec:bounds}
Instead of assuming full support, now we allow for limited support of the propensity score $P(x,Z)$, but we still require that it is an interval.
\begin{assumption}\label{limited_support}\textbf{Limited Support.} The support of $P(x,Z)$ is $\mathcal P_x=[\underline{p_x},\overline{p_x}]\subset[0,1]$.
\end{assumption}
\noindent Under Assumption \ref{limited_support}, and using \eqref{eq:obs_ps}, we have that the observed support of $P^*(X,Z)$ is
\begin{align*}
[\underline{p_x^*} , \overline{p_x^*} ]= [(1-\delta_x)\underline{p_x} + \delta_x\tilde P(x), (1-\delta_x)\overline{p_x} + \delta_x\tilde P(x)].
\end{align*}
Taking the difference we obtain that $\overline{p_x^*} - \underline{p_x^*}=(1-\delta_x)(\overline{p_x} -\underline{p_x} )$. Since $\overline{p_x} -\underline{p_x}\leq 1$, then $\overline{p_x^*} - \underline{p_x^*}\leq (1-\delta_x)$, so that a lower bound for $\delta_x$ is $\delta_x\geq 1-(\overline{p_x^*} - \underline{p_x^*})$.

In general, it is not possible to provide an upper bound for $\delta_x$. This is similar to the case of misclassification. Following that literature (see Assumption 4 in \cite{kedadni2021}, and references therein), we assume it is known that for some $\overline \delta_x$: $\delta_x\leq \overline \delta_x<1$. Thus, we can write $1-(\overline{p_x^*} - \underline{p_x^*})\leq \delta_x\leq \overline\delta_x.$ The correction factor in Examples \ref{example_late} and \ref{example_mprte} is $(1-\delta_x)$. Now, it bounded by $ 1-\overline \delta_x \leq 1-\delta_x\leq \overline{p_x^*} - \underline{p_x^*}$. Thus, we can bound both LATE and MPRTE using this:
\begin{align*}
(1-\overline\delta)\text{LATE}^*(x, P^*(x,z), P^*(x,z')) &\leq \text{LATE}(x, P(x,z), P(x,z'))\\&\leq  (\overline{p_x^*} - \underline{p_x^*})\text{LATE}^*(x, P^*(x,z), P^*(x,z')),
\end{align*}
and
\begin{align*}
(1-\overline\delta)\text{MPRTE}^*(x)\leq \text{MPRTE}(x) \leq  (\overline{p_x^*} - \underline{p_x^*})\text{MPRTE}^*(x).
\end{align*}
Naturally, if $\overline \delta_x$ is not known, we can only provide upper bounds.

Again, we stress that is not necessary to bound the MTE curve in the first place. Such a bound can be complicated to obtain since, by Lemma \ref{lemma:bias_mte}, $\delta_x$ enters in three different ways in the observed MTE curve.

\section{Misspecification as a weak instrument}\label{sec:weak_iv}
We can frame our model as the triangular scheme of \cite{stockstaiger1997} and consider a sequence $\left\{\delta_{x,n}\right\}_{n=1}^{\infty}$ such that $\lim_{n\to\infty}\delta_{x,n}=1$ at a certain rate as $n\to\infty$. Thus, as $n\to\infty$, the instrument becomes irrelevant in the model. A possible indicator of the presence of a large value of $\delta_{x,n}$ can be the average derivative of the observed propensity score. This equals an attenuated version of the average derivative of the true propensity score. For a given value of $\delta_{x,n}$, by equation \eqref{eq:der_pz}, we have
\begin{align*}
E\left[ \frac{\partial P^*(x,Z)}{\partial z}  \right] = (1-\delta_{x,n})E\left[ \frac{\partial P(x,Z)}{\partial z}  \right]
\end{align*}
Thus a ``small'' value can be an indication that $\delta_{x,n}$ is close to 1. This is similar to a first stage regression in the linear model. We take the derivative with respect to $z$ to get rid of the propensity score that does not respond to $Z$. We average, because this likely to be a non-linear expression. Thus, $(1-\delta_{x,n})$ can be thought of as the counterpart of $C/\sqrt T$ in the notation of \cite{stockstaiger1997}. Indeed, define
\begin{align*}
Cov_x(Z,D^*):=E[ZD^*|X=x] -E[Z|X=x]E[D^*|X=x]. 
\end{align*}
We have
\begin{align*}
E[ZD^*|X=x] &= E[ZSD|X=x] + E[Z(1-S)\tilde D|X=x]\\
&=E[ZSD|X=x] + E[Z|X=x]E[(1-S)\tilde D|X=x]
\end{align*}
and
\begin{align*}
E[D^*|X=x] = E[SD|X=x] + E[(1-S)\tilde D|X=x]
\end{align*}
Thus,
\begin{align*}
Cov_x(Z,D^*)&= E[ZSD|X=x] - E[Z|X=x]E[SD|X=x]\\
&+ E[Z|X=x]E[(1-S)\tilde D|X=x] -E[Z|X=x]E[(1-S)\tilde D|X=x] \\
&=Cov_x(Z,SD)
\end{align*}
which is the covariance between the instrument and treatment status for the responders with $X=x$. To see the role of the rate at which $\delta_{x,n}$ converges to 1, suppose for a second that we know the functional form of $P^*(x,Z)$, and we estimate the average derivative using a sample mean:
\begin{align*}
\hat E\left[ \frac{\partial P^*(x,Z)}{\partial z}  \right]  = \frac{1}{n}\sum_{i=1}^n \frac{\partial P^*(x,Z_i)}{\partial z} =(1-\delta_{x,n}) \frac{1}{n}\sum_{i=1}^n \frac{\partial P(x,Z_i)}{\partial z}
\end{align*}
Then
\begin{align*}
\hat E\left[ \frac{\partial P^*(x,Z)}{\partial z}  \right]  - E\left[ \frac{\partial P^*(x,Z)}{\partial z}  \right] = (1-\delta_{x,n})\left ( \frac{1}{n}\sum_{i=1}^n \frac{\partial P(x,Z_i)}{\partial z}-E\left[ \frac{\partial P(x,Z)}{\partial z}  \right]\right)
\end{align*}

In order to investigate possible discontinuities in the limiting distributions, we follow \cite{kuersteiner2002}, and we let $(1-\delta_{x,n})=n^{\nu_{x}}$, for $\nu_{x}<0$. We obtain
\begin{align*}
\hat E\left[ \frac{\partial P^*(X,Z)}{\partial z}  \right]  - E\left[ \frac{\partial P^*(X,Z)}{\partial z}  \right] = O_p(n^{\nu_{x}-1/2}).
\end{align*}
Then, we obtain a degenerate limit:
\begin{align*}
\sqrt n\left(\hat E\left[ \frac{\partial P^*(X,Z)}{\partial z}  \right]  - E\left[ \frac{\partial P^*(X,Z)}{\partial z}  \right]\right) = o_p(1)
\end{align*}

Now consider the MPRTE. Recall that, by Example \ref{example_mprte}, under the full support guaranteed by Assumption \ref{full_support}, 
\begin{align*}
n^{\nu_{x}}\text{MPRTE}^*(x) =\text{MPRTE}(x).
\end{align*}
Assume that, if $\delta_x=0$, there exists $\hat{\text{MPRTE}}(x)$, a $\sqrt n$-consistent estimator of $\text{MPRTE}(x)$ such that
\begin{align*}
\hat{\text{MPRTE}}^*(x)-\text{MPRTE}^*(x) = n^{-\nu_{x}}\left (\hat{\text{MPRTE}}(x) - \text{MPRTE}(x) \right).
\end{align*}
Thus, if $\nu_{x}=-1/2$, then $\hat{\text{MPRTE}}^*(x)$ does not converge in probability. In future work, we will use these results to construct confidence intervals for the parameters of interest.

\section{Conclusion}\label{sec:conclusion}
In this paper we use the MTE framework to model a proportion of individuals who do not respond to the incentives of the instrumental variable. We show that in the special case where the observed propensity score is fully supported on the unit interval, i) the CATE is automatically identified regardless of the non-responders, and ii) we can identify the proportion of non-responders and use it to recover the MTE curve, and we can recover any parameter associated with it. We show that for some parameters, such as LATE and MPRTE, it is even possible to bypass the recovery of the MTE curve, and directly recover these parameters. Moreover, if the propensity has limited support, we find bounds for the LATE, the MPRTE, and the MTE curve. When we let the proportion of non-responders approach 1 at a certain rate, the framework resembles that of weak instruments. In future research we hope to leverage the results in this literature to construct valid confidence intervals for the MTE curve and related parameters.

\bibliographystyle{econometrica}
\bibliography{bibliography}

\end{document}